\documentclass[twocolumn,prl,superscriptaddress]{revtex4}
\usepackage[latin9]{inputenc}
\usepackage{color}
\usepackage{amsmath}
\usepackage{amsthm}
\usepackage{amssymb}
\usepackage{graphicx}
\usepackage{physics}
\usepackage{float}
\usepackage{bbm}

\usepackage{tikz}
\usetikzlibrary{quantikz}
\usepackage{algorithm}  
\usepackage{algorithmicx}  
\usepackage{algpseudocode}  

\usepackage[unicode=true,
 bookmarks=true,bookmarksnumbered=false,bookmarksopen=false,
 breaklinks=false,pdfborder={0 0 1},backref=false,colorlinks=true]
 {hyperref}
\hypersetup{
 linkcolor=magenta, urlcolor=blue, citecolor=blue, pdfstartview={FitH}, hyperfootnotes=false, unicode=true}
 \usepackage{cleveref}

\makeatletter
\@ifundefined{textcolor}{}
{%
 \definecolor{BLACK}{gray}{0}
 \definecolor{WHITE}{gray}{1}
 \definecolor{RED}{rgb}{1,0,0}
 \definecolor{GREEN}{rgb}{0,1,0}
 \definecolor{BLUE}{rgb}{0,0,1}
 \definecolor{CYAN}{cmyk}{1,0,0,0}
 \definecolor{MAGENTA}{cmyk}{0,1,0,0}
 \definecolor{YELLOW}{cmyk}{0,0,1,0}
}

\usepackage{amsfonts,amsthm}
\usepackage{tabularx}
\usepackage{dcolumn}
\usepackage{bm}
\usepackage{graphicx}
\usepackage{epstopdf}
\usepackage{times}

\newtheorem{theorem}{Theorem}
\newtheorem{proposition}{Proposition}
\newtheorem{definition}{Definition}
\newtheorem{lemma}{Lemma}

\hypersetup{urlcolor=blue}

\def\II{{\mathcal I}}
\def\OO{{\mathcal O}}

\def\NN{{\mathcal N}}

\def\LL{{\mathcal L}}

\newcommand{\og}{\overline{G}}

\newcommand{\x}{\bm{x}}
\newcommand{\y}{\bm{y}}
\newcommand{\z}{\bm{z}}

\makeatother

\begin{document}

\title{Quantum-Classical Separations in Shallow-Circuit-Based Learning with and without Noises}

\author{Zhihan Zhang}
\thanks{These authors contributed equally to this work.}
\affiliation{Center for Quantum Information, IIIS, Tsinghua University, Beijing 100084, China}
\author{Weiyuan Gong}
\thanks{These authors contributed equally to this work.}
\affiliation{Center for Quantum Information, IIIS, Tsinghua University, Beijing 100084, China}
\author{Weikang Li}
\affiliation{Center for Quantum Information, IIIS, Tsinghua University, Beijing 100084, China}
\author{Dong-Ling Deng}
\email{dldeng@tsinghua.edu.cn}
\affiliation{Center for Quantum Information, IIIS, Tsinghua University, Beijing 100084, China}
\affiliation{Shanghai Qi Zhi Institute, 41st Floor, AI Tower, No. 701 Yunjin Road, Xuhui District, Shanghai 200232, China}
\affiliation{Hefei National Laboratory, Hefei 230088, China}

\begin{abstract}
We study quantum-classical separations between classical and quantum supervised learning models based on constant depth (i.e., shallow) circuits, in scenarios with and without noises. We construct a classification problem defined by a noiseless shallow quantum circuit and rigorously prove that any classical neural network with bounded connectivity requires logarithmic depth to output correctly with a larger-than-exponentially-small probability. This \textit{unconditional near-optimal} quantum-classical separation originates from the quantum nonlocality property that distinguishes quantum circuits from their classical counterparts. We further derive the noise thresholds for demonstrating such a separation on near-term quantum devices under the depolarization noise model. We prove that this separation will persist if the noise strength is upper bounded by an inverse polynomial with respect to the system size, and vanish if the noise strength is greater than an inverse polylogarithmic function. In addition, for quantum devices with constant noise strength, we prove that \textit{no} super-polynomial classical-quantum separation exists for any classification task defined by shallow Clifford circuits, independent of the structures of the circuits that specify the learning models.  
\end{abstract}
\maketitle

Quantum machine learning studies the interplay between machine learning and quantum physics \cite{Biamonte2017Quantum,Sarma2019Machine,Arunachalam2017Guest,Dunjko2018Machine,Ciliberto2017Quantum,Carleo2019Machine}. In recent years, a number of quantum learning algorithms have been proposed~\cite{Harrow2009Quantum,Lloyd2014Quantum,Dunjko2016QuantumEnhanced,Amin2018Quantum,Gao2018Quantum,Lloyd2018Quantum,Hu2019Quantum,Schuld2019Quantum,Liu2021Rigorous,Huang2021Information,Rebentrost2014Quantum,Bouland2021Noise,Gao2022Enhancing}, which may offer potential quantum advantages over their classical counterparts.  However, most of these algorithms either depend on a complexity assumption for the rigorous proof of the quantum advantage~\cite{Liu2021Rigorous}, or are out of the limited capabilities for noisy near-term quantum devices~\cite{Harrow2009Quantum,Lloyd2014Quantum,Lloyd2018Quantum,Rebentrost2014Quantum}. On the other hand, the difficulty of constructing fault-tolerant and deep quantum circuits motivates the study of machine learning using constant-depth (shallow) quantum circuits. Most of these algorithms assume classical access to data and use variational circuits for learning. Notable examples along this direction include machine learning based on variational quantum algorithms~\cite{Cerezo2021Variational,Mitarai2018Quantum,Grant2018Hierarchical,Schuld2020Circuit,Farhi2014Quantum,Farhi2018Classification,Li2023Ensemblelearning,Havlivcek2019Supervised}, which train a parameterized shallow quantum circuit with classical optimizers. For these experimental-friendly algorithms, however, there is no rigorous proof showing that they have genuine advantages over classical algorithms. These two points raise the essential problem of whether we can rigorously prove an unconditional quantum advantage in machine learning feasible for near-term quantum devices~\cite{Preskill2018Quantum}. Here, we establish a rigorous quantum-classical separation for the supervised learning model with noiseless shallow quantum circuits (as shown in Fig.~\ref{fig:illu} (a)) and pin down the noise regimes where the quantum-classical separation exists or disappears for shallow-circuit-based learning.

\begin{figure*}[t]
    \includegraphics[width=0.99\textwidth]{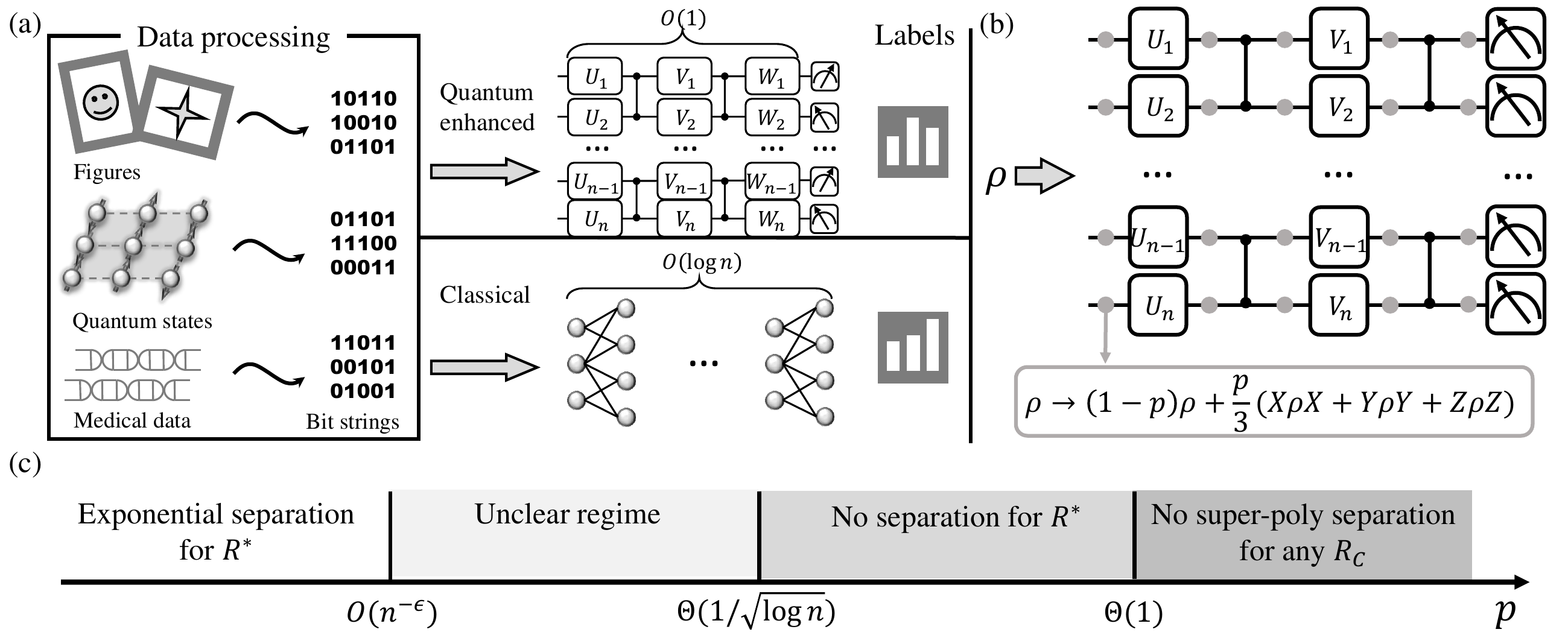}
    \caption{(a) An illustration of classical and quantum supervised learning models using shallow quantum circuits and neural networks with bounded connectivity. The data samples are first encoded into bit strings $\x$. These bit strings are input into a variational quantum circuit in the quantum-enhanced scenario and a classical neural network in the classical scenario. Each neuron in the classical neural network has bounded connectivity to the neurons in the previous layer. The outputs for two classifiers are distributions on all possible labels encoded in bit strings $\y$. (b) An implementation of a noisy quantum circuit, where an identical amount $p$ of depolarizing noise is added to each qubit at each step. (c) The characterization of noise rate regimes for the existence and absence of quantum-classical separation in learning classification tasks defined by shallow quantum circuits. Here, classification tasks $R^*$ and $R_C$ are relations defined in \Cref{thm:IdealSep} and \Cref{thm:NoisyBound}, respectively. The classification task $R^*$ can provide an exponential separation in the prediction accuracy when the noise rate is bounded above by $O(n^{-\epsilon})$, while such separation will vanish in the noise regime of $p>\Omega(1/\sqrt{\log n})$. In addition, any classification defined by a shallow Clifford circuit possesses no super-polynomial quantum-classical separation when the quantum circuit is implemented with a constant noise rate.}
    \label{fig:illu}
\end{figure*}

Recently, a rigorous and robust quantum advantage in supervised learning has been proved in Ref.~\cite{Liu2021Rigorous}, which relies on the assumption of the computational hardness in solving the discrete logarithm problem. Despite this exciting progress, an unconditional quantum-classical separation in learning is still lacking. In the quantum computing literature, a different paradigm to rigorously demonstrate quantum advantage was proposed in Ref.~\cite{Bravyi2018Quantum} to focus on shallow circuits, which is more feasible with noisy devices. It compared the circuit depth required for quantum and classical circuits to solve a relation problem, obtaining a separation originating from the classical hardness of simulating the intrinsic nonlocality property of quantum mechanics. In particular, it is proved that a shallow quantum circuit can solve a relation problem such that any classical circuits with bounded fan-in gates solving this problem with a moderate probability requires a logarithmic depth. Later works extended and amplified this result into an exponential separation in success probability~\cite{Gall2018Average,Watts2019Exponential,Coudron2021Trading,Aasnaess2022Comparing}. Other studies further showed that it is possible to preserve the separation from noises~\cite{Bravyi2020Quantum,Grier2021Interactive}. Very recently, a separation with similar intuition was proposed for sampling problems with shallow circuits~\cite{Watts2023Unconditional}. For classical neural networks, many convolution layers, ReLU layers, and pooling layers satisfy the local connectivity assumptions (i.e., bounded connectivity). Therefore, these unconditional separations provide the intuition that the nonlocality of quantum mechanics can provide an unconditional quantum advantage in machine learning tasks as well.

In this paper, we investigate the separations between shallow variational quantum circuits and classical neural networks for supervised learning tasks. For noiseless quantum devices, we prove that there exists a classification task defined by the input bits and output distributions of a shallow quantum circuit such that any classical neural network with bounded connectivity requires logarithmic depth to output correctly even with exponentially small probability $\exp(-O(n^{1-\epsilon}))$, where $\epsilon$ is a positive constant and $n$ is the system size. This indicates an \textit{unconditional near-optimal} separation between the representation power of shallow-circuit-based quantum classifiers and shallow-network-based classical neural networks. This quantum-classical separation originates from the classical hardness of simulating Bell nonlocality~\cite{Brunner2014Bell}. In addition, we consider a scenario with noises as well (Fig.~\ref{fig:illu} (b) and (c)). We show that when the noise rate is bounded above by $O(n^{-\epsilon})$, the quantum-classical separation will persist. On the contrary, if the noise rate exceeds $\Omega(1/\sqrt{\log n})$, this separation will no longer exist. In addition, we prove that any classification task defined by a noiseless shallow Clifford circuit provides \textit{no} super-polynomial quantum-classical separation if implemented on quantum devices with a constant noise rate.

\textit{Notations.---}Consider a generalized classification task in the context of supervised learning, where we assign multiple labels $\y_1,\cdots,\y_K\in\OO$ to an input data sample $\x\in\II$, with $\II$ and $\OO$ denoting the set of all possible samples and labels. Without loss of generality, we encode the samples and labels into bit strings of $n$ and $m$ bits, respectively. In essence, such a classification task can be described by a \textit{relation} $R:\{0,1\}^n\times\{0,1\}^m\to\{0,1\}$, such that a label $\y$ is a legitimate label for an input sample $\x$ if and only if $R(\x,\y)=1$. For a supervised learning algorithm, we can also describe the inputs and outputs through a relation (called a hypothesis relation) $R_h:\{0,1\}^n\times\{0,1\}^m\to\{0,1\}$. For an input sample $\x$, we assign $R_h(\x,\y)=1$ if $\y$ is a possible output for the learning algorithm and $R_h(\x,\y)=0$ otherwise. Given an input $\x$, the learning algorithm outputs a label $\y\in\{0,1\}^m$ with probability $\Pr[\y|\x,R_h]$. To evaluate the performance of a learning algorithm, we use a loss function defined as the probability of the learning algorithm outputting a wrong label averaged uniformly over all possible input strings $\x\in\{0,1\}^n$:
\begin{align}\label{eq:LossProb}
\LL_P(R,R_h)=1-\frac{1}{2^n}\sum_{\x,\y}\Pr[\y|\x,R_h]\mathbbm{1}_{R(\x,\y)=1},
\end{align}
where $\mathbbm{1}_{A}$ takes value $1$ when the argument $A$ is satisfied and $0$ otherwise. We say a relation can be represented by a neural network or a quantum circuit if the loss function can be reduced to zero by choosing a set of suitable variational parameters. 

\textit{Separations for noiseless quantum circuits.---}We compare the representation power of shallow quantum circuits and shallow classical neural networks equipped with neurons of bounded connectivity. We start with variational quantum circuits without noise. In particular, we prove the following theorem, which shows a near-optimal quantum-classical separation between these two types of classifiers.
\begin{theorem}\label{thm:IdealSep}
There exists a classification task described by a relation $R^*:\{0,1\}^n\times\{0,1\}^m\to\{0,1\}$ such that a constant-depth parametrized variational quantum circuit with single- and two-qubit gates  can represent this relation (i.e., with zero loss). However, any classical neural network with neurons of bounded connectivity that can represent a hypothesis relation $R_h$ with
\begin{align}
\LL_P(R^*,R_h)\leq 1-\exp(-\gamma n^{1-\epsilon})
\end{align}
for some constant $\gamma>0$ requires depth at least $\Omega(\epsilon\log n)$ for any positive constant $\epsilon<1$.
\end{theorem}
\begin{proof}
We provide the main idea here. The detailed proof is
technically involved and thus left to the Supplementary Materials. Specifically,we prove the equivalence of the classical circuit model and neural network model with bounded connectivity, and a proposition indicating that there exists an elementary classification task $R_e:\{0,1\}^{n_e}\times\{0,1\}^{m_e}\to\{0,1\}$ such that there exists a constant-depth quantum circuit that outputs a correct label $\y$ for any given sample $\x$ with unity probability. Whereas, any (randomized) classical neural network with bounded connectivity outputs correctly with probability at least $1-\alpha$ for some constant $\alpha$ requires depth $\Omega(\log n_e)$. We then exploit a direct product of $t$ copies of the relation $R_e$. We show that, by choosing large enough $t$, one can find a relation $R^*=R_e^{\times t}$ with $n=n_et$ and $m=m_et$ to amplify the quantum-classical loss function separation to $0$ versus $1-\exp(-O(n^{1-\epsilon}))$ for any positive constant $\epsilon<1$. This completes the proof for \Cref{thm:IdealSep}.
\end{proof}
The separation shown in \Cref{thm:IdealSep} can be attributed to the difficulty of simulating quantum nonlocality  using classical algorithms~\cite{Barrett2007Modeling}. More concretely, we consider supervised learning the measurement outcomes when measuring each qubit of a well-designed graph state~\cite{Hein2004Multparty} on the X-basis or the Y-basis. The graph states contain long cycles, requiring entanglement between distant sites in classical simulation. As a result, any classical neural network with bounded connectivity requires logarithmic depth to connect these distant wires. 
Similar ideas have been exploited to demonstrate the computational power separation between shallow classical and quantum circuits~\cite{Bravyi2018Quantum,Bravyi2020Quantum,Gall2018Average,Coudron2021Trading,Watts2019Exponential}. In particular, an analogous result for the special case $\epsilon=\frac{1}{2}$ is proved in~\cite{Gall2018Average}. 

\Cref{thm:IdealSep} shows an \textit{exponential} quantum-classical separation in the success probability of outputting a correct label for the classification task. This result indicates a quantum-classical separation in the representation power between classical neural networks with bounded connectivity and variational quantum circuit in supervised learning. Shallow quantum circuits are more powerful in representing complicated relations and functions. In addition, this separation also implies a constant-versus-logarithmic time separation in the inference stage of the classification task, between quantum classifiers based on constant-depth variational circuits and classical classifiers based on neural networks with bounded connectivity.

Furthermore, we rigorously prove that there exists a classical neural network of depth $c_0\log n$ with bounded connectivity and fine-tuned parameters that can reduce the loss function to $1-\exp(-O(n/\sqrt{\log n}))$ for any relation defined by a two-dimensional shallow quantum circuit (see Supplementary Materials Sec. III). Therefore, the quantum-classical separation shown in \Cref{thm:IdealSep} is near-optimal. In particular, to prove the optimality we explicitly propose an algorithm. This algorithm divides the input qubits into blocks of $O(\sqrt{\log n})\times O(\sqrt{\log n})$ qubits. We show that for shallow quantum circuits with single- and two-qubit gates, a classical simulation can output correctly if there are no inter-block gates. We then prove that the number of qubits in the input string affected by inter-block gates is bounded above by $O(n/\sqrt{\log n})$. Therefore, this classical simulation can output a correct string for all the qubits except the affected qubits. Finally, we apply a random guessing strategy on these affected qubits, and the probability of outputting a correct string is thus $\exp(-O(n/\sqrt{\log n}))$.

\Cref{thm:IdealSep} can be extended straightforwardly to loss functions other than that defined in Eq.~\eqref{eq:LossProb}. For example, we consider the loss function defined on $R^*$ and $R_h$ based on Kullback-Leibler (KL) divergence~\cite{Kullback1951Information}. In the Supplementary Materials, we show that relation $R^*$ in \Cref{thm:IdealSep} outputs according to an equal distribution on all possible strings $\y$. If we compute the KL divergence between $R^*$ and the probability distribution output from a learning algorithm, the loss function for the variational quantum circuit can be reduced to zero. However, any classical neural network with bounded connectivity and loss function smaller than $O(n^{1-\epsilon})$ requires depth $\Omega(\epsilon\log n)$. 

\textit{Noisy quantum circuits.---}Short-term quantum devices suffer from noises \cite{Preskill2018Quantum}. As a result, a question of both theoretical and experimental interest is whether the quantum-classical separation discussed above persists for noisy quantum circuits. It is shown that the existence of noise leads to the ineffectiveness of deep quantum circuits caused by the barren plateau phenomenon~\cite{Wang2021Can,Cerezo2021Cost,Wang2021Noise,Anschuetz2022Quantum}. Yet, whether shallow circuits lose their quantum advantage in the presence of noise remains largely unexplored. To address this important question, we focus on the depolarization noise model ~\cite{Nielsen2010Quantum} $\rho\rightarrow(1-p)\rho+\frac{p}{3}(X\rho X+Y\rho Y+Z\rho Z)$ for the shallow quantum circuit as an example. We prove that the variational quantum circuits in \Cref{thm:IdealSep} suffering from noise rate $p$  can only output correctly with probability $\exp(-\Theta(pn))$. Therefore, we conclude that the quantum-classical separation in \Cref{thm:IdealSep} remains when $p=O(n^{-\epsilon})$. However, the noise rate on most near-term quantum devices without any error mitigation or correction is typically assumed to be $p=O(1)$~\cite{Chen2022Complexity,Aharonov2022Polynomial,Preskill2018Quantum,Liu2021Benchmarking}. In such a scenario, we rigorously prove that the exponential separation in \Cref{thm:IdealSep} vanishes. In the above discussion, we have proposed a classical algorithm using logarithmic depth to output a correct label of a shallow quantum circuit with probability at least $\exp(-O(n/\sqrt{\log n}))$ when proving the near-optimality of the noiseless quantum-classical separation. By comparing this result and the performance bound for noisy shallow quantum circuits, the largest noise threshold for this quantum-classical separation will not exceed $p=O(1/\sqrt{\log n})$, which decreases as the system size increases. An illustrative figure for the noise regimes where the quantum-classical separation exists or disappears for shallow-circuit-based learning is provided as \Cref{fig:illu} (b). We leave the details for the derivation of upper and lower bound on the noise threshold for demonstrating \Cref{thm:IdealSep} on noisy devices in the Supplementary Materials.

We further find that the nonexistence of quantum-classical separation in the context of supervised learning using shallow circuits can be extended to more general cases. In particular, we derive the following theorem showing that there is no super-polynomial quantum-classical separation between noisy shallow variational quantum circuits and classical learning algorithms using shallow neural networks or circuits.  

\begin{theorem}\label{thm:NoisyBound}
For any classification task corresponding to a relation $R_C:\{0,1\}^n\times\{0,1\}^m\rightarrow\{0,1\}$ that can be defined  by a constant-depth (classically-controlled) Clifford circuit, no quantum circuit undergoing the depolarization noise $\rho\to(1-p)\rho+p(X\rho X+Y\rho Y+Z\rho Z)/3$ with constant noise rate $p$ can achieve a loss less than
\begin{align}
1-(1-\LL_P(R_C,R_r))^{\log_2\frac{1}{1-\frac{2}{3}p}},
\end{align}
where $R_r$ is the relation represented by a classical neural network that outputs uniformly randomly over $\y\in\{0,1\}^m$.
\end{theorem}

\begin{proof}
We sketch the major steps here and leave the technical details to Supplemental Material Sec. IV. We adapt a lemma from Ref.~\cite{Dehaene2003Clifford,Nest2008Classical} and prove that the output of the quantum circuit defining $R_C$ uniformly distributes over an affine subspace of $\mathbb{Z}_2^m$. Therefore, there exists a subset $\y_S$ of output bits such that when the input bits and the output bits outside of $\y_S$ are fixed, only one assignment of the bits in $\y_S$ can satisfy $R_C$, meaning that the probability of success for a classical random guessing algorithm is $(\frac{1}{2})^{|\y_S|}$. Similarly, for noisy quantum circuits, we consider the noise in the layer exactly before measuring the bits in $\y_S$ conditioned on all other noises. We conclude that the output label for a noisy quantum circuit for $R_C$ is correct with probability at most $(1-\frac{2}{3}p)^{|\y_S|}$. This completes the proof for this theorem.
\end{proof}

We recall that the shallow variational quantum circuit in \Cref{thm:IdealSep} contains non-Clifford gates. However, it can be proved that this classification task can also be represented by a classically-controlled Clifford quantum circuit~\cite{Gall2018Average}. Thus, the result in \Cref{thm:NoisyBound} also works for the classification task in \Cref{thm:IdealSep}. As $p$ is a constant, $\log_2(1/(1-2p/3))$ is a constant. Therefore, \Cref{thm:NoisyBound} shows that the success probability for a noisy shallow quantum circuit to output a correct label for any classification task defined by a noiseless shallow Clifford circuit scale at most polynomially as that of a classical algorithm that randomly guesses the labels. Since random guesses can be accomplished by a classical neural network of constant depth, the separation in \Cref{thm:IdealSep} does not exist here. \Cref{thm:NoisyBound} further indicates that we cannot expect a near-term experimental demonstration of quantum advantage only using classification tasks characterized by a shallow Clifford circuit.

It is worthwhile to remark that the relation considered in \Cref{thm:NoisyBound} can be defined by a shallow noiseless Clifford quantum circuit, which is different from the relations considered in Ref.~\cite{Bravyi2020Quantum,Grier2021Interactive}. There, an implicit error correction technique is introduced to guarantee that the quantum-classical separation persists for noisy circuits. Such technique takes a problem with a separation between classical circuits and noiseless quantum circuits and encodes the problem in error-correcting codes. Then the error syndromes are measured without carrying out the error correction and the classical circuit is required to do the same. This yields a separation between shallow classical circuits and noisy quantum circuits of constant noise rate. Therefore, the resulting relation can be represented but not defined by a shallow quantum circuit as the quantum circuit is allowed to make a mistake with a moderate probability. Whereas, in our scenario, the output distribution of the relation is assumed to be exactly the same with some quantum circuits, yielding a more restricted class of relations. We conclude that this kind of restriction is vital for the nonexistence of the separation in Ref.~\cite{Bravyi2020Quantum,Grier2021Interactive}. 

We also clarify the differences between \Cref{thm:NoisyBound} and the recent results in Ref.~\cite{Aharonov2022Polynomial} in the absence of quantum advantage on constant-rate noisy random circuit sampling. Ref.~\cite{Aharonov2022Polynomial} considered the task aiming at approximating the output distribution of a random quantum circuit within a small total variational difference. It is widely believed that classically simulating the output distributions of a (noisy) random circuit of logarithmic depth exactly is difficult~\cite{Aaronson2011Computational,Bouland2021Noise,Boixo2018Characterizing,Bouland2019Complexity,Movassagh2019Quantum,Aaronson2017Complexity,Aaronson2020Classical}, but Ref.~\cite{Aharonov2022Polynomial} explicitly propose an efficient algorithm to approximate the outputs distribution of noisy random circuits with polynomial complexity. While this result gives evidence against the quantum advantage of random circuit sampling with a constant noise rate, \Cref{thm:NoisyBound} focus on the classification task represented by quantum circuits of depth $O(1)$. The random circuits are assumed to satisfy the anti-concentration condition, which requires a circuit depth of $\Omega(\log n)$, while we consider shallow circuits of constant depth. It is hard to classically simulate the result of random circuits, but the output distribution of the shallow quantum circuit considered in this paper can be efficiently simulated by classical algorithms~\cite{Napp2022Efficient,Wang2022Possibilistic}. To construct a quantum-classical separation in this regime, we compare the representation power of shallow quantum circuits to that of shallow classical neural networks and show the quantum-classical separation. We prove that classification tasks defined by shallow Clifford circuits can not demonstrate any super-polynomial quantum advantage on near-term devices with constant noise rates.

\textit{Outlook.---}Many questions remain and are worth further investigation. For instance, our discussion in this paper focused on the relation exactly defined by a shallow quantum circuit. For quantum devices with a constant noisy rate, it is still possible to find a quantum-classical separation in the success probability for relations not exactly defined by a quantum circuit. How to explicitly construct such a hard classification task to rigorously obtain a noise-tolerant quantum-classical separation without relying on quantum error correction is crucial for demonstrating quantum advantages with near-term quantum devices. In practice, many classical neural networks bear unbounded connectivity. How to extend our results to this scenario demands further exploration. In addition, how to extend our results to unsupervised learning and reinforcement learning scenarios remains unclear. Whether we can transfer our separation in representation power to a separation in sample complexity or computational complexity is an important future direction as well. Finally, it would be of both fundamental and practical importance to carry out an experiment to demonstrate the unconditional quantum-classical separation proved in \Cref{thm:IdealSep}. This would be a far-reaching step towards practical quantum learning supremacy in the future. 

\textit{Acknowledgement}.---We thank Sitan Chen, Tongyang Li, Dong Yuan, Zidu Liu, Qi Ye, Xingjian Li, Zhide Lu, and Xun Gao for helpful discussions. This work is supported by the National Natural Science Foundation of China (Grants No. 12075128 and No. T2225008) and the Shanghai Qi Zhi Institute.

\bibliographystyle{apsrev4-1-title}
\bibliography{QAdvSuperLearn}

\clearpage
\newpage 
\onecolumngrid
\makeatletter
\setcounter{figure}{0}
\setcounter{table}{0}
\setcounter{equation}{0}
\setcounter{theorem}{0}
\setcounter{lemma}{0}
\setcounter{definition}{0}
\renewcommand{\thefigure}{S\@arabic\c@figure}
\renewcommand{\thetable}{S\@arabic\c@table}
\renewcommand{\theequation}{S\@arabic\c@equation}
\renewcommand{\thetheorem}{S\@arabic\c@theorem}
\renewcommand{\thelemma}{S\@arabic\c@lemma}
\renewcommand{\thedefinition}{S\@arabic\c@definition}

\begin{center} 
    {\large \bf Supplementary Materials for:\\ Quantum-Classical Separations in Shallow-Circuit-Based Learning with and without Noise}
\end{center} 

In the Supplementary Materials, we provide more details about the proofs for the two theorems in the main text. We provide supplementary notes on the variational quantum circuit and its applications in classification. For the noiseless quantum-classical separation demonstrated in Theorem 1 in the main text, we derive the noise rate regime for this separation to persist or vanish.

\section{Variational Quantum Circuits and Classifiers}\label{sec:VQC}
In this section, we provide a brief introduction to variational quantum circuits and their applications in classification tasks.
As indicated by the name, variational quantum circuits are quantum circuits that include variational parameters inside them \cite{Cerezo2021Variational}.
For example, it is popular to use the rotation angles of some single-qubit rotation gates to encode the variational parameters as shown below:
$$
\begin{quantikz}
& \gate{R_x(\theta)} & \qw
\end{quantikz}
= e^{-i \frac{\theta}{2} \hat{\sigma}_x},\;
\begin{quantikz}
& \gate{R_y(\theta)} & \qw
\end{quantikz}
= e^{-i \frac{\theta}{2} \hat{\sigma}_y},\;
\begin{quantikz}
& \gate{R_z(\theta)} & \qw
\end{quantikz}
= e^{-i \frac{\theta}{2} \hat{\sigma}_z}.
$$
Besides, in quantum approximate optimization algorithms, the evolution time of the fixed Hamiltonian can also be utilized as a variational parameter.
To entangle different qubits, we can use two-qubit controlled gates (such as controlled-NOT gates and controlled-Z gates), the time evolution of a multi-qubit Hamiltonian, and many other alternatives.
For practical tasks, the choice of the circuit ansatz should be made according to the experimental condition and the specific task.
Considering the limited size and non-negligible experimental noise, algorithms based on variational quantum circuits are suitable candidates for demonstrating early-stage applications on near-term quantum devices \cite{Preskill2018Quantum}.

\begin{figure*}[h]
    \centering
    \includegraphics[width=0.95\textwidth]{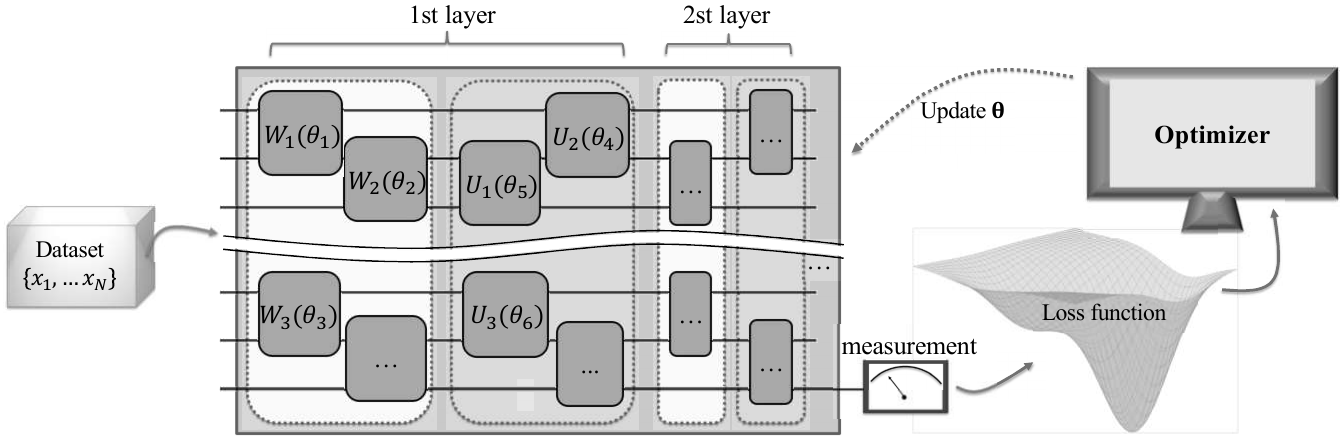}
    \caption{Schematics of supervised learning with variational quantum circuits. After feeding the training samples into the model and obtaining the outputs, the optimizer will update the variational parameters in order to minimize the loss function.}
    \label{fig:vqc}
\end{figure*}

We can utilize variational quantum circuits to handle some optimization-based tasks. Here, we focus on classification tasks. First, in the supervised learning setting, there is a training dataset with size $M$
\begin{equation}
\mathcal{D}=\left\{\left({\x}_1, y_1\right),\left({\x}_2, y_2\right), \ldots,\left({\x}_M, y_M\right)\right\},
\end{equation}
which our variational quantum circuit will learn from.
For concreteness, before starting the learning procedure, we need to define the way in which we encode the training data $\x_i$ into the variational quantum circuit model and the way in which we read the output. Then, to learn from the dataset $\mathcal{D}$, what we do is to feed the samples $(\x_i,y_i)$ into the model and tune the variational parameters such that the output is close to the label $y_i$.

To achieve this, we usually need a loss function to measure the distance between the current output value and the ground truth label. Let $h(\x;\boldsymbol{\theta})$ be the hypothesis function, where $\x$ and $\boldsymbol{\theta}$ denote the input and the variational parameters, respectively.
Then, we use $\mathcal{L}(h(\x; \boldsymbol{\theta}), y)$ to denote the loss function. $\mathcal{L}(h({\x}; \boldsymbol{\theta}), y)$ will be small if $h(\x; \boldsymbol{\theta})$ is close to $y$ and large otherwise.
For example, suppose we write the output value and the label in the form of one-hot encoding $\mathbf{g}$ and $\mathbf{a}$, respectively.
Two commonly used loss functions, i.e., the mean squared error and cross-entropy, can be expressed as
\begin{equation}
\mathcal{L}_{MSE}(\mathbf{g}, \mathbf{a})  =\sum_k\left(a_k-g_k\right)^2,
\end{equation}
\begin{equation}
\mathcal{L}_{CE}(\mathbf{g}, \mathbf{a})  =-\sum_k a_k \log g_k.
\end{equation}
Thus, the learning task can be converted to the minimization task for the loss function.
To this end, various gradient-descent-based and gradient-free methods have been proposed \cite{Mitarai2018Quantum,Havlivcek2019Supervised}. The basic procedure is summarized as shown in Fig.~\ref{fig:vqc}.
After the training process, the variational quantum classifier can be utilized to make predictions on new samples.

In this work, we consider a more general case of classification tasks.
As mentioned in the main text,
we assign multiple labels $\y_1,\cdots,\y_K\in\OO$ to an input data sample $\x\in\II$, and without loss of generality, we can encode the samples and the labels into bit strings of $n$ and $m$ bits, respectively.
In this scenario, 
we can use a relation $R:\{0,1\}^n\times\{0,1\}^m\to\{0,1\}$ to describe the classification criterion:
a label $\y$ is legitimate for an input sample $\x$ if and only if $R(\x,\y)=1$. 
In supervised learning,
we can also use a relation, i.e., a hypothesis relation $R_h:\{0,1\}^n\times\{0,1\}^m\to\{0,1\}$, to describe the input-output pairs.
More concretely, given an input sample $\x$, we assign $R_h(\x,\y)=1$ to a possible output $\y$ for the model and $R_h(\x,\y)=0$ otherwise.
In this generalized case, we can accordingly change the loss function, e.g., we can directly use the probability of the model making a wrong prediction (defined in Eq.~(1) in the main text)
or use the Kullback-Leibler (KL) divergence (defined in Eq.~\ref{eq:LossKL}).

\section{Proof and Extensions of Theorem 1}\label{sec:PfThm1}
Here, we provide the proof for Theorem 1 in the main text and its extensions to different loss functions. In addition to the ones in the main text, we first give more
notations and definitions to describe the construction for the relation $R^*$ in Theorem 1. We consider a graph state~\cite{Hein2004Multparty} $\ket{G}$ defined on a graph $G(V,E)$, where $V$ is the set of vertex and $E$ is the set of edges. It is a $\abs{V}$-qubit state defined by the following process. We start with the initial state $\otimes_{u\in V}\ket{0}_{Q_u}$, where $Q_u$ represents the single-qubit register for the qubit associated with vertex $u$. We then apply a Hadamard gate on each register and a control-Z gate on $(Q_u,Q_v)$ for each pair of adjacent vertices $(u,v)\in E$. This graph state is stabilized under operators (i.e., unchanged if we apply these operators)
\begin{align}
\pi_u=X_u\times\bigotimes_{v\in N(u)}Z_v,
\end{align}
where $X_u$ is the Pauli-X operator applied to the register $Q_u$, $Z_v$ is the Pauli-Z operator applied to the register $Q_v$, and $N(u)$ is the set of all adjacent vertices of $u$. 

Given an undirected graph $G(V,E)$, we denote $\overline{G}$ the extended graph of $G$ with $\abs{V}+\abs{E}$ vertices and $2\abs{E}$ edges obtained from inserting a vertex at each edge of $G$. We denote $V^*$ the set of the inserted vertices so the vertex set of $\overline{G}$ is $V\cup V^*$. Now, we are ready to recap the construction of the elementary relation $R_e:\{0,1\}^{n_e}\times\{0,1\}^{m_e}\to\{0,1\}$ in Ref.~\cite{Gall2018Average}. 

Given an even positive integer $d$, we construct a graph $G_d$ with vertices set $V_d=V_1\cup V_2$ defined as shown in Fig. 4 of Ref.~\cite{Gall2018Average}. We consider a $d^3\times d^3$ square grid and denote the set of vertices as $V_1$. We can observe that $\abs{V_1}=\Theta(d^6)$. The vertices are composed of $d^4$ identical and contiguous square regions each containing a $d\times d$ grid. We call such a $d\times d$ region a box. Next, we place a vertex at the center of each $1\times 1$ grid in the boxes. We connect each added vertex with the four corners of the corresponding grid. We denote $V_2$ as the set of the new vertices added and $\abs{V_2}=d^4(d-1)^2=\Theta(d^6)$. We label each vertex by $u_{ij}$ for $i,j\in\{1,...,k\}$, where $k=d^2(d-1)$, the index $j$ represents the vertical position, and the index $i$ represents the horizontal position. This completes the description of $G_d$.

The graph $\og_d$ is the extended graph of $G_d$. Let $V^*$ denote the inserted vertices and $\abs{V^*}=2d^3(d^3-1)+4d^4(d-1)^2$. Denote $\overline{V}_d=V_1\cup V_2\cup V^*$ to be the set of all vertices in $\og_d$ and let
\begin{align}
m_e=\abs{\overline{V}_d}=\Theta(d^6).
\end{align}

We then define the relation $R_e$ based on the graph state on $\og_d$. Given a matrix $A\in\{0,1\}^{k\times k}$, we consider the process $P_d(A)$ described as follows. We first construct the graph state over $\og_d$. For each $u_{ij}\in V_2$, we measure the $Q_{u_{ij}}$ in the $X$-basis if $A_{ij}=0$ and in the $Y$-basis if $A_{ij}=1$. For each vertex $u\in V_1\cup V^*$, we measure the $Q_u$ in the $X$-basis. In this process, we perform a measurement on each node of $\og_d$, outputting a single classical bit each. We represent the output by a bit string of length $m_e$ using an arbitrary ordering of the ${m_e}$ nodes of $\og_d$. We denote $\Lambda_d(A)\subseteq\{0,1\}^{n_e}$ to be the set of all the strings that occur with non-zero probability in $P_d(A)$. The formal definition of the relation is given below.
\begin{definition}\label{def:GallThm}
For any even positive integer $d$ and a given matrix $A\in\{0,1\}^{k\times k}$ as input, where $k=d^2(d-1)$, we compute a string from $\Lambda_d(A)$. Since $\abs{\Lambda_d(A)}>1$, there is more than one valid output. For any input $A\in\{0,1\}^{k\times k}$ and output $\z\in\{0,1\}^{m_e}$ the relation $R_e(A,\z)=1$ for the following set
\begin{align}
\{(A,\z)|A\in\{0,1\}^{k\times k} \text{ and } \z\in\Lambda_d(A)\}\subseteq\{0,1\}^{k\times k}\times\{0,1\}^{m_e}.
\end{align}
By setting $n_e=k^2$ and rewrite $\{0,1\}^{k\times k}$ with $\{0,1\}^{n_e}$, we can regard $R_e$ as a subset of $\{0,1\}^{n_e}\times\{0,1\}^{m_e}\to\{0,1\}$.
\end{definition}

As long as $3m_e^{1/7}<d-2$, the following theorem holds.
\begin{lemma}[Theorem 2, Ref.~\cite{Gall2018Average}]\label{thm:GallThm}
There exists a relation $R_e:\{0,1\}^{n_e}\times\{0,1\}^{m_e}\to\{0,1\}$ for which the following two assertions hold
\begin{itemize}
    \item A constant-depth quantum circuit with single- and two-qubit gates can output a string $\y\in\{0,1\}^{m_e}$ for any input $\x\in\{0,1\}^{n_e}$ with probability $1$ such that $R_e(\x,\y)=1$.
    \item There exists a constant $\alpha>0$ such that any classical circuit $C$ with bounded fan-in gates satisfying
    \begin{align}
    \frac{1}{2^{n_e}}\sum_{\x\in\{0,1\}^{n_e}} \Pr[R_e(\x,C(\x))=1]\geq1-\alpha
    \end{align}
    has depth $\Omega(\log n_e)$.
\end{itemize}
\end{lemma}

The proof of the theorem is given in Ref.~\cite{Gall2018Average}. We only recap the intuition here. In the quantum case, the computational problem in \Cref{def:GallThm} can be solved by directly implementing the process $P_d(A)$, which can be done by a constant-depth quantum circuit as $\og_d$ has a constant maximum degree. 

For any classical circuit of sub-logarithmic depth with bounded fan-in gates, the light cones for each bit of the output for the circuit are limited by a constant. Simulating the outputs for the measurement on a graph state requires long-range interactions. However, the limited size of light cones makes a large fraction of the bit pairs in $A$ not related to each other with high probability. Therefore, the circuit cannot output a string in $\Lambda_d(A)$ with a high probability for a non-negligible fraction of the inputs $A$.

We now consider an amplification of the elementary relation $R_e$.  For any relation $R_e\subseteq\{0,1\}^{n_e}\times\{0,1\}^{m_e}\to\{0,1\}$ with positive integers $n_e$ and $m_e$, we construct the following relation. Given $t>0$ input strings $\x_1,\ldots,\x_t\in\{0,1\}^{n_e}$, the relation outputs one element from the set $R_e(\x_1)\times\ldots\times R_e(\x_t)$. This new relation problem corresponds to the direct product of $t$ copies of the original relation $R_e$. For simplicity, we denote the relation $R_e^{\times t}:\{0,1\}^{n_et}\times\{0,1\}^{m_et}\to\{0,1\}$. If $R_e$ cannot be solved with an averaged success probability larger than $1-\alpha$ using sub-logarithmic classical circuits, $R_e^{\times t}$ cannot be solved with an averaged success probability larger than $(1-\alpha)^{t'}$ for some $t'<t$ by any circuit of the same depth. The core idea is to extract at least $t'$ copies of $R_e$ on which the circuit acts independently from $R_e^{\times t}$. The formal result is the following repetition theorem from Ref.~\cite{Gall2018Average}.
\begin{lemma}[Theorem 10, Ref.~\cite{Gall2018Average}]\label{thm:GallRepThm}
Let $R_e\subseteq\{0,1\}^{n_e}\times\{0,1\}^{m_e}\to\{0,1\}$ be a relation for which the following property holds for some $c>0$ and $\alpha\in[0,1]$: any $n_e$-input $m_e$-output (randomized) classical circuit $C$ with bounded fan-in gates and depth $c\log n_e$ can success with probability at most
\begin{align}
\frac{1}{2^m}\sum_{\x\in\{0,1\}^{n_e}}\Pr[R_e(\x,C(\x))=1]<1-\alpha.
\end{align}
Let $t\geq 6m_en_e^c+2$ be a positive integer. Then any $(mt)$-input $(nt)$-output (randomized) classical circuit $C'$ with bounded fan-in gates and depth $c\log n_e=O(\log (n/t)), n=n_et$ satisfies
\begin{align}
\frac{1}{2^{n_et}}\sum_{\x'\in\{0,1\}^{n_et}}\Pr[R_e^{\times t}(\x',C'(\x'))=1]<(1-\alpha)^{t/(6m_en_e^c+2)}.
\end{align}
\end{lemma}

The proof of \Cref{thm:GallRepThm} is provided in Ref.~\cite{Gall2018Average}. Here, we consider applying \Cref{thm:GallRepThm} on the relation $R_e$ we constructed in \Cref{thm:GallThm}. For this relation, we have $n_e=\Theta(d^6)$ and $m_e=\Theta(d^6)$. We take the positive integer $t=\lceil(6m_en_e^{1/8}+2)^{l}\rceil$ for $l=\frac19(\frac{17}{\epsilon}-8)$, where $\epsilon<1$ is a positive constant. The input and the output size for $R^*=R_e^{\times t}$ are $n=n_et$ and $m=m_et$. It is straightforward to observe that $t>n_e^{9l/8}$, which implies $t>n^{\frac{9l}{9l+8}}$. \Cref{thm:GallThm} and \Cref{thm:GallRepThm} indicates that there exist constants $c>0$ and $\alpha>0$ such that any $n$-input $m$-output (randomized) classical circuit $C'$ with bounded fan-in gates and depth at most $c\log n_e=O(\epsilon \log n)$ satisfies
\begin{align}
\frac{1}{2^n}\sum_{\x'\in\{0,1\}^n}\Pr[R^*(\x',C'(\x'))=1]<(1-\alpha)^{t^{(l-1)/l}}\leq(1-\alpha)^{n^{\frac{l-1}{l}\cdot\frac{9l}{9l+8}}}=(1-\alpha)^{n^{1-\epsilon}},
\end{align}
which leads to the following theorem.
\begin{proposition}\label{thm:CircuitThm}
There exists a relation $R^*:\{0,1\}^n\times\{0,1\}^m\to\{0,1\}$ for which the following two assertions hold
\begin{itemize}
    \item There exists a constant-depth quantum circuit with single- and two-qubit gates that for any input $\x\in\{0,1\}^n$ can output a string $\y\in\{0,1\}^m$ such that $R^*(\x,\y)=1$ for certain.
    \item Any classical circuit $C$ with bounded fan-in gates satisfying
    \begin{align}
    \frac{1}{2^n}\sum_{\x\in\{0,1\}^n} \Pr[R^*(\x,C(\x))=1]\geq\exp(-\gamma n^{1-\epsilon})
    \end{align}
    for some constant $\gamma>0$ has depth $\Omega(\log n_e)=\Omega(\epsilon\log n)$.
\end{itemize}
\end{proposition}

We then derive Theorem 1 in the main text based on \Cref{thm:CircuitThm}. We consider classical feedforward neural networks. In these neural networks, neurons are divided into different layers with the first layer known as the input layer, the last layer known as the output layer, and the intermediate layers known as hidden layers. Each neuron in this network is connected to some of the neurons in the previous layer. The value of the neuron depends only on the neurons it is connected to. We denote \textit{connectivity} as the number of neurons that are connected to a chosen neuron.

We observe that a neuron can be viewed as a gate. The connectivity for a neuron is the fan-in number for the gate. Therefore, \Cref{thm:CircuitThm} also holds for classical neural networks with bounded connectivity. The subtlety here is that while gates compute Boolean functions, neurons compute real-valued functions. But since the input and the output from the output layer are all Boolean, each output bit of the neural network is a (possibly random combination) of deterministic Boolean function that depends only on a restricted set of input bits. Since the proof for \Cref{thm:CircuitThm} only relies on such restricted causal dependencies, and a random combination of answers cannot perform better than the best deterministic answer, this subtlety does not pose a problem.

For the quantum circuit with measurement, we consider the variational quantum circuit mentioned in the previous section \ref{sec:VQC}. As the preparation procedure of the graph state in \Cref{thm:CircuitThm} only contains single- and two-qubit gates, we can use a classically-controlled quantum circuit with constant depth and Clifford gates to depict the relation. The difference between this classically-controlled quantum circuit and the (variational) quantum circuit model in \Cref{sec:VQC} is that classical input can change the parameters for gates. To transfer this model to a quantum circuit model, we first encode the classical bit string into a quantum state (for example, $0001\to\ket{0001}$ for $n=4$). Next, we use a controlled-R gate for
\begin{align}
R=\frac{1}{\sqrt{2}}\begin{pmatrix}1 & - i\\ 1 & i\end{pmatrix}
\end{align}
to control the original qubits using these input qubits to realize the base changes between measurements on the X- and Y-basis. In this case, the circuit model has a quantum input and outputs the same result as the classically-controlled quantum circuit. By picking some fine-tuned parameters, we can reduce the variational quantum circuit to this circuit. Therefore, we only require a variational quantum circuit with constant depth and fine-tuned parameters to represent the relation $R^*$. As mentioned in the main text, such a relation can be regarded as a hypothesis function in supervised learning. This completes the proof for Theorem 1 in the main text.

In the following part, we extend Theorem 1 in the main text to loss functions different from Eq.~(1) in the main text.  As mentioned in Ref.~\cite{Gall2018Average}, the probability distribution associated with $R^*$ in Theorem 1 over all the legitimate outputs $\y$ is an uniform distribution (see \Cref{sec:PfThm2} for a proof). Therefore, the probability distribution for quantum circuit $p_{\mathcal{Q}}(\y;\x)$ is also an equal distribution over all possible outcomes. We consider the output distribution $p_C(\y;\x)$ over $\y\in\{0,1\}^m$ using a classical neural network with bounded connectivity and depth at most $O(\epsilon\log n)$. Here, we choose the loss function for the neural network $C$ as the KL divergence between $p_C(\y;\x)$ and $p_{\mathcal{Q}}(\y;\x)$ (the ideal distribution) averaged over all $\x$
\begin{align}\label{eq:LossKL}
\mathcal{L}_{KL}(R_C,R^*)=\frac{1}{2^n}\sum_{\x}d_{KL}(p_{\mathcal{Q}}(\y;\x),p_C(\y;\x))=-\frac{1}{2^n}\sum_{\x}\sum_{\y}p_{\mathcal{Q}}(\y;\x)\log(\frac{p_C(\y;\x)}{p_{\mathcal{Q}}(\y;\x)}).
\end{align}
According to Theorem 1 in the main text, the probability of the neural network outputting even one correcting label $\y$ is bounded by $\exp(-\gamma n^{1-\epsilon})$. Thus the loss function is bounded below by
\begin{align}
\mathcal{L}_{KL}(R_C,R^*)&=-\frac{1}{2^n}\sum_{\x}\sum_{\y\in S(\x)}\frac{1}{\abs{S(\x)}}\log\left(\abs{S(\x)}p_C(\y;\x)\right)\\
&\geq-\frac{1}{2^n}\sum_{\x}\log\left(\sum_{u\in S(\x)}\frac{1}{|S(\x)|}|S(\x)|p_C(\y;\x)\right)\\
&\geq\frac{1}{2^n}\sum_{\x}\gamma n^{1-\epsilon}=\gamma n^{1-\epsilon},
\end{align}
where $S(\x)$ denotes all the $\y$ that is correct for input $\x$, i.e. $R^*(\x,\y)=1$. For the variational quantum circuit considered above, the loss is $0$.
We have thus shown that the quantum-classical separation can be extended to different loss functions that are commonly considered in (quantum) machine learning.

\section{Noise Thresholds for the quantum-classical Separation}\label{sec:NoiseCircLower}
Here, we consider the noise thresholds for the quantum-classical separation in Theorem 1 in the main text to hold under noisy quantum implementation in Fig. 1 (b) in the main text. We first consider the performance of variational quantum circuits for supervised learning tasks when we add a noise of strength $p$, i.e., a depolarization channel of the form
\begin{align}\label{eq:NoiseChannel}
\rho\to\NN(\rho)=(1-p)\rho+\frac{p}{3}(X\rho X+Y\rho Y+Z\rho Z).
\end{align}
after each qubit in each step. In particular, we show the following bounds on the loss function in Eq.~(1) in the main text when the quantum classifier we considered in Theorem 1 in the main text suffers from noise.
\begin{theorem}\label{thm:NoiseCircPerf}
We consider the relation $R^*$ and the constant-depth variational quantum circuit with fine-tuned parameters in Theorem 1. Suppose the noise follows the model in Eq.~\eqref{eq:NoiseChannel}, then the loss function between the noiseless relation $R^*:\{0,1\}^n\times\{0,1\}^m\to\{0,1\}$ and the relation $R_Q$ defined by the noisy shallow circuit is bounded by
\begin{align*}
\LL_P(R^*,R_Q)=1-\exp(-\Theta(pn)).
\end{align*}
\end{theorem}
\begin{proof}
We first prove the lower bound for the loss function as
\begin{align}
\LL_P(R^*,R_Q)=1-\exp(-O(pn)).
\end{align}
We consider the fidelity between the output of the noisy quantum circuit and the output of the noiseless case. As we have assumed that a depolarization channel occurs after each qubit in each layer, we just compute the contribution to the fidelity from the part $(1-p)\rho$ in each noise channel and ignore the contribution from the term $\frac{p}{3}(X\rho X+Y\rho Y+Z\rho Z)$. Therefore, the fidelity is bounded below by
\begin{align}
F\geq (1-p)^{nd}\sim\exp(-O(pnd))
\end{align}
for some constant depth $d$ as there are at most $nd$ noise channels after $O(nd)$ gates. As the noiseless quantum circuit exactly outputs the relation with zero loss, the loss for the noisy shallow quantum circuit is therefore bounded above by $1-\exp(-O(pn))$.

Next, we prove the lower bound for the loss function as
\begin{align}
\LL_P(R,R_Q)=1-\exp(-\Omega(pn)).
\end{align}
We consider Le Gall's construction of hard instance in proving \Cref{thm:GallThm}. It is defined by measurement outcomes on a graph state. By the definition of graph state, the output bits on each cycle in the graph are constrained by an equality~\cite{Gall2018Average}. The smallest cycle in the figure is a small triangle containing two vertices in $V_1$, one vertex in $V_2$, and three vertices in $V^*$. It is straightforward to observe that we can always choose $\Theta(n_e)$ \textit{disjoint} cycles. In the following, we consider the three vertices in $V^*$ which should satisfy theorem $5$ of Ref.~\cite{Gall2018Average}. We consider the errors before and after the \textit{last computational layer} of the quantum circuit separately. Without loss of generality, we can move the errors between the last computational layer and the measurements to the same qubit after the measurements. After such transformation, each depolarization channel becomes a binary symmetric channel with crossover probability $2p/3$ because only Pauli-X and Pauli-Y errors will flip the bits. Thus the errors that flip the bits of the output string happen with a $2p/3$ rate. We consider the two possible cases depending on the different measurement results:
\begin{itemize}
    \item If the measurement results satisfy the constraint in one of the cycles, then the output after the binary symmetric channels still satisfies the constraint with probability at most
    \begin{align}
    P_{\text{one-layer}}<\left(1-\frac{2p}{3}\right)^3+3\left(1-\frac{2p}{3}\right)\left(\frac{2p}{3}\right)^2\leq e^{-p/2}.
    \end{align}
    Since one or three errors simultaneously will lead to an error.
    \item If the measurement results do not satisfy the constraint in one of the cycles, then the output after the binary symmetric channels satisfies the constraint with probability at most
    \begin{align}
    P_{\text{one-layer}}<3\left(1-\frac{2p}{3}\right)^2\left(\frac{2p}{3}\right)+\left(\frac{2p}{3}\right)^3\leq e^{-p/2}.
    \end{align}
\end{itemize}
Therefore, in either case, the output of the noisy quantum circuit will satisfy the constraint with probability at most $\exp(-p/2)$ for each cycle. Since there are $\Theta(n_e)$ disjoint cycles and the constraints are necessary conditions for the output to be correct, the probability of a correct output for the noisy quantum circuit is upper bounded by $\exp(-\Omega(pn_e)))$ for the elementary relation. For the relation $R^*$ after amplification, the probability for a correct output is bounded by $\exp(-\Omega(pn)))$.
\end{proof}

\Cref{thm:NoiseCircPerf} demonstrates that the noisy implementation of the variational quantum circuit with fine-tuned parameters can reduce the loss function $\LL_P$ to $1-\exp(-\Theta(pn))$. This result indicates that at least with probability $\exp(-\Theta(pn))$ the quantum circuit can output a correct label. Comparing \Cref{thm:NoiseCircPerf} with Theorem 1 in the main text, we can observe that the depolarization noise will result in an exponential decay in the probability of outputting a correct label originated from the exponential decay in the fidelity. If the noise rate $p$ is bounded by
\begin{align}\label{eq:NoiseRateA}
p=O(n^{-\epsilon}),
\end{align}
for some constant $\epsilon>0$, we can make sure that the quantum-classical separation persists on noisy devices.

Eq.~\eqref{eq:NoiseRateA} gives the noise threshold for the quantum-classical separation to persist in Theorem 1. For a classical neural network with bounded connectivity of depth $O(\epsilon\log n)$, Theorem 1 shows that it can output a correct label for $R^*$ with probability no more than $\exp(-O(n^{1-\epsilon}))$ for a positive constant $\epsilon<1$. However, this is an upper bound on the performance of any classical neural network. It would be natural to ask whether we can prove a tighter lower bound on the performance of classical neural networks. Moreover, it is important to explore whether it is possible to improve the noise threshold in Eq.~\eqref{eq:NoiseRateA} to a constant one. In the following, we give a no-go theorem for keeping quantum advantage under a constant noise rate by proving the lower bound on the representation power for classical circuits based on classical simulation of Clifford $+$ T gates~\cite{Wang2022Possibilistic,Napp2022Efficient}. As a necessary assumption, we notice that the relation defined by a graph state can be described by a quantum circuit containing only Clifford and T gates with a constant depth. For this class of relation, we have the following theorem.
\begin{theorem}\label{thm:nogolower}
If the relation $R\in\{0,1\}^n\times\{0,1\}^m\to\{0,1\}$ can be defined by a two-dimensional shallow quantum circuit on a grid containing only Clifford gates and T gates, there always exists a classical circuit $C$ containing only bounded fan-in gates of depth $c\log n$ that can output a correct output for the relation with probability at least
\begin{align}
\Pr[R(\x,C(\x))=1]\geq\exp(-\gamma(c)O\left(\frac{n}{\sqrt{\log(n)}}\right))
\end{align}
for any input $\x\in\{0,1\}^n$, where $\gamma(c)$ is a function that only depends on $c$.
\end{theorem}

Before providing the proof of \Cref{thm:nogolower}, we consider the most trivial classical algorithm: randomly guessing the output for the relation. This classical algorithm has no requirement on the depth of the classical circuit or neural network. Notice that since there are in total $2^n$ output strings and at least one of them is a possible output of the quantum circuit, this algorithm will always output a correct output of the relation with probability at least $2^{-n}=\exp(-\Theta(n))$. Using this bound, we can expect a quantum advantage even with a constant noise rate according to Theorem 1 and \Cref{thm:NoiseCircPerf}. However, we have a classical neural network with bounded connectivity and logarithmic depth, it is natural to expect that we can perform better than randomly guessing the output. In the following, we give an affirmative answer to this intuition. To begin with, we introduce some basic notations and concepts. We follow the concepts used in Ref.~\cite{Wang2022Possibilistic}.
\begin{definition}
A quantum circuit $Q$ with $n$ input qubits and
measurement on $m$ output qubits in the computational basis defines a relation $R_{QC}\subseteq\{0,1\}^n\times\{0,1\}^m\to\{0,1\}$ by
\begin{align}
R_{QC}(\x,\y)=1\Leftrightarrow\bra{\y}Q\ket{\x}\neq0,
\end{align}
where $Q\ket{\x}$ is the output for $Q$ with input $\x\in\{0,1\}^n$. Let $C$ be a classical circuit and let $R$ be a relation. We say $C$ \textit{p-simulates} $R$ if:
\begin{align}
R(\x,C(\x))=1 
\end{align}
for any $\x\in\{0,1\}^n$.
\end{definition}
\begin{figure*}[t]
    \centering
    \includegraphics[width=0.99\textwidth]{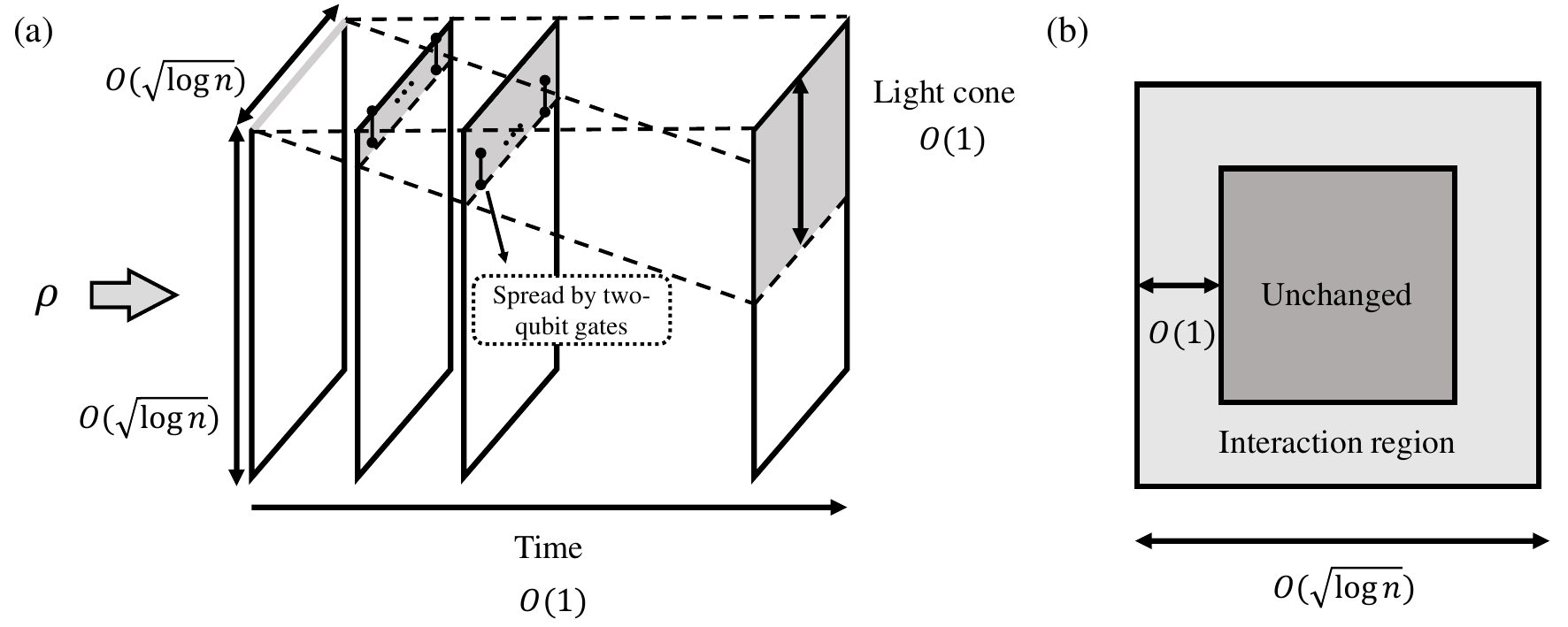}
    \caption{Depth complexity in simulating shallow Clifford $+$ T quantum circuit using classical circuits. As the quantum circuit only contains single- and two-qubit gates, the light cone of a gate can contain at most $O(1)$ qubits in the $\sqrt{c\log n/c_1}\times \sqrt{c\log n/c_1}$ block. Therefore, the interacting qubits on one border can at most affect the output qubits of size $\sqrt{c\log n/c_1}\times O(1)$ as shown in (a). As shown in (b), all the output qubits in the block can thus be divided into qubits (the outer region) interacting with the qubits outside the block and the qubits (the inner region) independent of qubits outside the block.}
    \label{fig:ClassSim}
\end{figure*}
Under this definition, Ref.~\cite{Wang2022Possibilistic} gives a construction for classical circuits using $2$-fan-in gates $\{$AND, OR,  NOT$\}$ that can p-simulate the outputs of a (shallow) quantum circuit.
\begin{lemma}[Theorem 1, Ref.~\cite{Wang2022Possibilistic}]\label{lem:ClassSimQ}
Any $n$-qubit quantum circuit $Q$ of depth $d$
containing Clifford gates and $t$ T gates can be p-simulated by a classical circuit $C$ of depth $O(d+t)$ containing only $\{$AND, OR, NOT $\}$ gates.
\end{lemma}

As we assume that $Q$ has only constant depth $d$, this shallow quantum circuit can be p-simulated by a classical circuit with bounded fan-in of depth $O(d+n)=O(n)$ as there are at most $t=nd=O(n)$ T gates. Without loss of generality, we assume that the coefficient here is $c_1$, i.e., we can always use a classical circuit with $c_1n$ depth containing only bounded fan-in gates. We are then ready for proving \Cref{thm:nogolower}.

\begin{proof}
Without loss of generality, we assume that the two-dimensional quantum circuit is defined on a $\sqrt{n}\times\sqrt{n}$ grid. It is worthwhile to mention that our result also carries over to other qubit patterns represented by planar graphs such that the interaction strength between a subsystem and the rest of the system is only related to the boundary of the subsystem instead of the subsystem size. We then divide the grids into $\sqrt{c\log n/c_1}\times\sqrt{c\log n/c_1}$ blocks. There are in total $\frac{nc_1}{c\log n}$ blocks. According to \Cref{lem:ClassSimQ}, there exists a classical circuit of depth $c\log n$ that can p-simulate each block ignoring the inter-block gates.

If there are no inter-block gates, then this classical circuit can output one correct output for certain, regardless of the input $\x\in\{0,1\}^n$. Now, we consider the effect of inter-block gates as shown in Fig.~\ref{fig:ClassSim}. For each block, there are $4\sqrt{c\log n/c_1}$ border qubits that can be affected directly by inter-block gates. As the quantum circuit only contains single- and two-qubit Clifford + T gates, all the qubits that can be affected by inter-block gates are bounded by the size of the light cone, which is $2^d$. As there are at most $4d\sqrt{c\log n/c_1}$ inter-block gates in each block, the qubits that can be affected by these gates are bounded by $4d2^d\sqrt{c\log n/c_1}$. Thus, the number of qubits that can be affected in $n$ qubits is
\begin{align}
\frac{nc_1}{c\log n}\cdot4d2^d\sqrt{c\log n/c_1}=4d2^d\frac{\sqrt{c_1}}{\sqrt{c}}\cdot\frac{n}{\sqrt{\log n}}.
\end{align}
If we randomly guess one output string for these qubits, then the classical algorithm can p-simulate the quantum circuit with probability at least
\begin{align}
P_{\text{success}}=2^{-4d2^d\frac{\sqrt{c_1}}{\sqrt{c}}\cdot\frac{n}{\sqrt{\log n}}}=\exp(-4\log 2~ d2^d\frac{\sqrt{c_1}}{\sqrt{c}}\cdot\frac{n}{\sqrt{\log n}}).
\end{align}
We then choose $\gamma(c)=4\log 2~d2^d\frac{\sqrt{c_1}}{\sqrt{c}}$ and finish the proof for \Cref{thm:nogolower}.
\end{proof}

This theorem provides an algorithmic lower bound on the probability for classical circuits using bounded fan-in gates in simulating shallow quantum circuits containing Clifford $+$ T gates. By combining the performance bound for noisy quantum circuits in \Cref{thm:NoiseCircPerf} and the above lower bound, we can conclude that the optimal possible noise rate is bounded by
\begin{align}
p=O\left(\frac{1}{\sqrt{\log n}}\right)
\end{align}
if we want to demonstrate the quantum-classical separation using the relation $R^*$ in Theorem 1 in the main text.

\section{Proof of Theorem 2}\label{sec:PfThm2}
In this section, we provide the proof for Theorem 2.
We begin with the following lemma.
\begin{lemma}[Ref.~\cite{Dehaene2003Clifford,Nest2008Classical}]\label{lem:CliffordState}
For any Clifford circuit $\mathcal{C}$ with $m$ qubits,
\begin{align}
\mathcal{C}|0\rangle\propto\sum_{y\in G}i^{l(y)}(-1)^{q(y)}|y\rangle,
\end{align}
where $G\subseteq \mathbb{Z}_2^m$ is an affine subspace, $l(y)$ is a linear function on $\mathbb{Z}_2^m$ and $q(y)$ is a quadratic function on $\mathbb{Z}_2^m$.
\end{lemma}
\Cref{lem:CliffordState} was proved in Ref.~\cite{Dehaene2003Clifford,Nest2008Classical}. Here, we only require that the support of $\mathcal{C}|0\rangle$ is an affine subspace, so we present another proof as the starting point.

Consider the stabilizer group $\mathcal{S}\subset\mathcal{P}^m$ of $\mathcal{C}|0\rangle$.
The elements in $\mathcal{S}$ with the form of a tensor product of the $I$ and $Z$ operators on each qubit form a subgroup $G\subseteq\mathcal{S}$.
There is a natural injective homomorphism $\phi$ from $G$ to $\mathbb{Z}_2^{m+1}$: for $1\leq i\leq m$, the $i$th bit of $\phi(g)$ is $1$ if $g$ has an $Z$ operator acting on the $i$th bit and $0$ otherwise.
The $m+1$th bit of $\phi(g)$ is $1$ if $g$ has a minus sign, and $0$ otherwise.
Suppose the rank of $G$ is $r\leq m$, then $\phi(G)$ is an $r$ dimensional subspace of $\mathbb{Z}_2^{m+1}$.
Suppose $g_1,\ldots,g_r$ is a basis for it.
The probability of obtaining the output $\y\in\mathbb{Z}_2^m$ for the input $\x$ is
\begin{align}
\Pr[\y|\x]=\langle\y|\frac{1}{2^m}\sum_{s\in\mathcal{S}}s|\y\rangle=\frac{1}{2^m}\sum_{g\in G}\langle \y|g|\y\rangle=\frac{1}{2^m}\sum_{g\in G}(-1)^{\phi(g)\cdot\tilde{\y}}.
\end{align}
Here, $\tilde{\y}$ is the $m+1$-dimensional vector obtained by appending $\y$ with a $1$. The summation above is nonzero if and only if $g_i\cdot\tilde{\y}=0$ for $i=1,\ldots,r$, in which case it takes the value $2^{r-m}$. Therefore, the output distribution for input $\x$ after the stabilizer circuit $\mathcal{C}(\x)$ is a uniform distribution over the $2^{m-r}$ possible output strings $\y$. That is, any possible output string $\y$ should satisfy a system of linear equations $C\y=b$, where $C$ is a $r\times m$ matrix, and $b$ is an $r$-dimensional vector for an appropriate $b$.

Now, we are ready to prove Theorem 2 in the main text. Since the relation is defined by a (classically-controlled) Clifford circuit, when the input $\x$ is fixed, the correct output is determined by a Clifford circuit. Thus we have a corresponding $C(\x)$. As $C(\x)$ is full row rank, we can find $r(\x)$ linearly independent columns of $C(\x)$. We denote the set of corresponding output bits $\y_S\subseteq\y$ for an output string $\y$. When the output bits outside $\y_S$ are fixed, only one possible assignment for the bits in $\y_S$ can satisfy $C(\x)\y=b(\x)$. Thus, when we are randomly guessing the output bits conditioned on the guesses of bits outside of $\y_S$, the probability of guessing the bits in $\y_S$ to form a correct output for the relation is $P_{\text{guess}}(\x)=2^{-r(\x)}$. Therefore, we have found that the classical random guessing algorithm which requires only constant circuit depth can output correctly with probability $2^{-r(\x)}$ given the input is $\x$. Average over all input $\x$, the average success probability is $P_{\text{guess}}=\frac{1}{2^n}\sum_{\x}2^{-r(\x)}$.

We then consider the performance of noisy quantum circuits with constant noise rates. Similar to the techniques we exploit in the previous section, we consider the effect of the depolarizing channels immediately before the measurement of bits in $S$ in the layer before the measurement. We move these depolarization channels to the same places after the measurement, yielding binary symmetric channels with crossover probability $2p/3$. Consider the $r(\x)$ binary symmetric channels conditioned on all the measurement results, only one flipping configuration will result in an output allowed by the relation. This configuration happens with probability
\begin{align}
P_{\text{noisy}}(\x)\leq\left(1-\frac{2p}{3}\right)^{r(\x)}=P_{\text{guess}}(\x)^{\log_2\frac{1}{1-\frac{2p}{3}}}. 
\end{align}
Since $0\leq p\leq\frac{3}{4}$, $0\leq\log_2\frac{1}{1-\frac{2p}{3}}\leq1$, average over all input $\x$, we get the average success probability
\begin{align}
P_{\text{noisy}}=\frac{1}{2^n}\sum_{\x}P_{\text{noisy}}(\x)\leq\frac{1}{2^n}\sum_{\x}P_{\text{guess}}(\x)^{\log_2\frac{1}{1-\frac{2p}{3}}}\leq\left(\frac{1}{2^n}\sum_{\x}P_{\text{guess}}(\x)\right)^{\log_2\frac{1}{1-\frac{2p}{3}}}=P_{\text{guess}}^{\log_2\frac{1}{1-\frac{2p}{3}}}.
\end{align}
By transferring this result into the context of supervised learning through a similar technique to Theorem 1, we complete the proof for Theorem 2 in the main text.
\end{document}